\documentclass{siamltex}

\usepackage{amssymb}
\usepackage{amsmath}
\usepackage{color}
\usepackage{graphicx}
\usepackage{enumitem}

\usepackage[letterpaper, margin={1in}]{geometry}

\newcommand{\subrefb}[2]{ \eqref{#1}\raisebox{-1mm}{\scriptsize #2} }

\DeclareMathOperator*{\argmin}{arg\,min}

\title{Euler elastica as a $\Gamma$-Limit of discrete bending energies of one-dimensional chains of atoms}

\author{
    Malena I.~Espa\~ nol\footnotemark[1]
\and
    Dmitry Golovaty\footnotemark[1]
\and
    J. Patrick Wilber\footnotemark[1]}

\begin{document}
\maketitle

\renewcommand{\thefootnote}{\fnsymbol{footnote}}
\footnotetext[1]{
    Department of Mathematics,
    University of Akron,
    Akron,
    OH 44325, USA.}

\begin{abstract}
This work is motivated by discrete-to-continuum modeling of the
mechanics of a graphene sheet, which is a single-atom thick
macromolecule of carbon atoms covalently bonded to form a hexagonal
lattice.  The strong covalent bonding makes the sheet essentially
inextensible and gives the sheet a resistance to bending.  We study a
one-dimensional atomistic model that describes the cross-section of a
graphene sheet as a collection of rigid links connected by torsional springs. 
$\Gamma$-convergence is used to
rigorously justify an upscaling procedure 
for the discrete bending energy of
the atomistic model.  Our result establishes that as the bond length in
the atomistic model goes to 0, the bending energies $\Gamma$-converge
to Euler's elastica.
\end{abstract}


\section{Introduction}
This work is motivated by discrete-to-continuum modeling of the mechanics
of graphene.  A graphene sheet is a single-atom thick macromolecule of
carbon atoms arranged in a hexagonal lattice.  Graphene has been
intensively studied since 2004 when individual graphene sheets were
first isolated by Geim and Novoselov using mechanical exfoliation, a
Nobel-prize-winning achievement \cite{novoselov2004electric}.
Graphene has exceptional physical properties and yields insights on
fundamental physics of two-dimensional materials.  It is also of
interest as a basic building block for other extensively studied
carbon nanostructures (CNS), including single-walled and multi-walled
carbon nanotubes, as well as bilayers and stacks of graphene sheets.

Interactions between carbon atoms in a CNS are fundamental for
explaining the arrangement and relative orientations of the
nanostucture's constituent graphene layers.  Each atom on a given
graphene sheet is covalently bonded to its three nearest neighbors to
form a hexagonal lattice.  This strong covalent bonding makes the
sheet essentially inextensible and gives the sheet a resistance to
bending.  Atoms on a sheet interact with atoms on nearby sheets by
relatively weak van der Waals forces.  This interaction defines an
equilibrium distance between pairs of weakly interacting atoms, and deviations from this
distance cost energy.  The weak interaction energy is minimized when
nearby lattices are in registry.  The lattices of two adjacent sheets
adjust or shift to allow a typical atom on one sheet to be close to
the equilibrium distance from its several nearest neighbors on the
adjacent sheet.

Registry effects are significant for understanding the mechanical
behavior and equilibrium configurations of various CNS.  For example,
registry effects lead to polygonization of multi-walled carbon
nanotubes of large diameter \cite{golovaty2008continuum}.  To
understand the influence of registry effects on CNS at the macroscopic
level, it is important to construct continuum models that retain
atomistic information.  

In \cite{golovaty2008continuum}, the authors derive a continuum theory
of multi-walled carbon nanotubes by upscaling a simple one-dimensional
atomistic model that takes into account both strong covalent bonds
between the atoms in a graphene sheet and weak bonds between the atoms
in adjacent sheets.  Part of this model is based on upscaling a
resistance to bending described atomistically.  The resulting
continuum bending energy takes the form of the classical Euler's elastica
model \cite{antman2006nonlinear}.  Given a sufficiently smooth curve $\mathcal C\subset\mathbb{R}^2$, the Euler's elastica model assigns to $\mathcal C$ the bending energy $\int_{\mathcal C}\kappa^2$, where $\kappa$ is the curvature of $\mathcal C$. Within a larger effort to provide a rigorous justification for
the procedure used in \cite{golovaty2008continuum}, a first step is to
consider upscaling from a discrete to a continuum bending energy of a
single graphene sheet.  Hence, in this paper we use
$\Gamma$-convergence \cite{braides2002gamma} to rigorously justify the
upscaling procedure used in \cite{golovaty2008continuum} for the
bending energy of a chain of atoms.

Our work is related to that of Bruckstein et al.\@
\cite{bruckstein107,bruckstein_igl}, who studied discrete
approximations of the classical elastica model motivated by problems
in image processing.  In \cite{bruckstein107,bruckstein_igl}, the discrete energies were
defined on piecewise-affine curves and assumed to depend on the exterior
angles between the straight segments of the curve.  The authors considered
$\Gamma$-convergence for several related families of discrete energy
functionals defined on the space of rectifiable planar curves of
finite total absolute curvature.  The convergence in the space of
rectifiable curves was defined in the sense of Fr\'echet distance.  An
advantage of working in the space of rectifiable curves of finite
total absolute curvature is that it contains both smooth and piecewise-affine
curves.  The limiting energy functionals in
\cite{bruckstein107,bruckstein_igl} are essentially the
$L^{\alpha}$-norm of the curvature, where $\alpha\geq1$.

The main distinguishing feature between the motivation in this paper and that in \cite{bruckstein107,bruckstein_igl} is that we are approximating a discrete chain of atoms by a continuum curve, while in \cite{bruckstein107,bruckstein_igl} the goal is to approximate a continuum curve by a polygon. As a result, our discrete model is determined by the physics of the problem, while in \cite{bruckstein107,bruckstein_igl} the discrete framework is determined by the convenience of the approximation. In a nondimensional setting, our model
represents a cross-section of a graphene sheet as a chain of atoms
in which all links connecting the atoms have equal length $\varepsilon>0$ while the total
length of the chain is $1$.  One may think of this chain as a polygon in the plane. The parameter $\varepsilon$ representing
the interatomic bond length is assumed to be small.  To mimic the
situation in \cite{golovaty2008continuum}, we assume that the chain is
closed, thus describing a graphene sheet rolled into a carbon
nanotube.  This assumption, however, is not essential to our analysis
and can easily be removed.  As $\varepsilon\to0$, the corresponding
chains converge to a curve on the plane that is a continuum description of a cross-section of the nanotube. 

Instead of working with curves directly, as in
\cite{bruckstein107,bruckstein_igl}, we represent each arc-length-parametrized curve by its
corresponding angle function. A discrete atomic chain then is described by a piecewise-constant function whose values are the angles between the links of the
chain and the $x$-axis. The angles remain constant on each successive
subinterval of the length $\varepsilon$ of $[0,1]$. To pass from the discrete to a continuum description, we assume that 
as $\varepsilon\to0$, the sequence of angle functions converges in an appropriate
sense to a limiting function defined on $[0,1]$. 

We define the bending energy of the chain as a function of the angles between the adjacent
links of the chain and thus of the increments of the angle
function. In the limit $\varepsilon\to 0$, the bending energy becomes a function of the derivative of the limiting angle function. Then, if we expect that the bending energy
reduces to Euler's elastica as $\varepsilon\to 0$, the
limiting angle function must have a square integrable derivative. From this we would conclude that the
limiting angle function is smoother than the discrete angle functions that converge to it. 

To prove $\Gamma$-convergence of the discrete bending energies to Euler's elastica, we follow the strategy employed in \cite{braides2002gamma} and
replace the piecewise-constant functions with auxiliary piecewise-affine functions that have the same discrete-level energy and that
belong to the same space as the limiting angle function. The
principal difficulty in proving $\Gamma$-convergence is the
construction of the recovery sequence of the angle functions. In particular, here we need to design a recovery sequence that satisfies the constraints  that the piecewise-affine curves must be closed and have unit length. Because of the length constraint our construction is more complicated than its analog in \cite{bruckstein107,bruckstein_igl}.

This paper is organized as follows.  In Section~\ref{s1}, we formulate
our discrete model of the bending energy of an atomistic chain.  The
following section introduces the continuum analogue of this discrete
model and sets up the appropriate function spaces for our
$\Gamma$-convergence result.  Section~\ref{s3} contains this result,
Theorem~4.3, and its proof.

\section{Discrete formulation}\label{s1}

Given a small $\varepsilon>0$ such that $N_\varepsilon:=1/\varepsilon\in\mathbb N$, let $\left\{\bar
r_i^\varepsilon\right\}_{i=1}^{N_\varepsilon} \subset \mathbb{R}^2$ be
an ordered set of position vectors for $N_\varepsilon$ points in the plane
such that $\left|\bar r_{i+1}^\varepsilon-\bar
r_i^\varepsilon\right|=\varepsilon$ for 
$i=1,\ldots,N_\varepsilon-1$ and $\left|\bar
r_{N_\varepsilon}^\varepsilon-\bar
r_1^\varepsilon\right|=\varepsilon$.  We refer to $\left\{\bar
r_i^\varepsilon\right\}_{i=1}^{N_\varepsilon}$ as a {\em chain}
$\mathcal C^\varepsilon$.
As Fig. \ref{fig:1} shows, $\mathcal C^\varepsilon$ can be associated with a piecewise-affine curve $\mathcal C_\varepsilon$ in $\mathbb{R}^2$ by connecting the
consecutive points in $\left\{\bar
r_i^\varepsilon\right\}_{i=1}^{N_\varepsilon}$. This piecewise-affine curve has
length 1. Notice that
we use superscript $\varepsilon$ for the discrete chain and subscript
$\varepsilon$ for the associated curve. 

\begin{figure}[htb]
\centering
\includegraphics[height=1.5in]{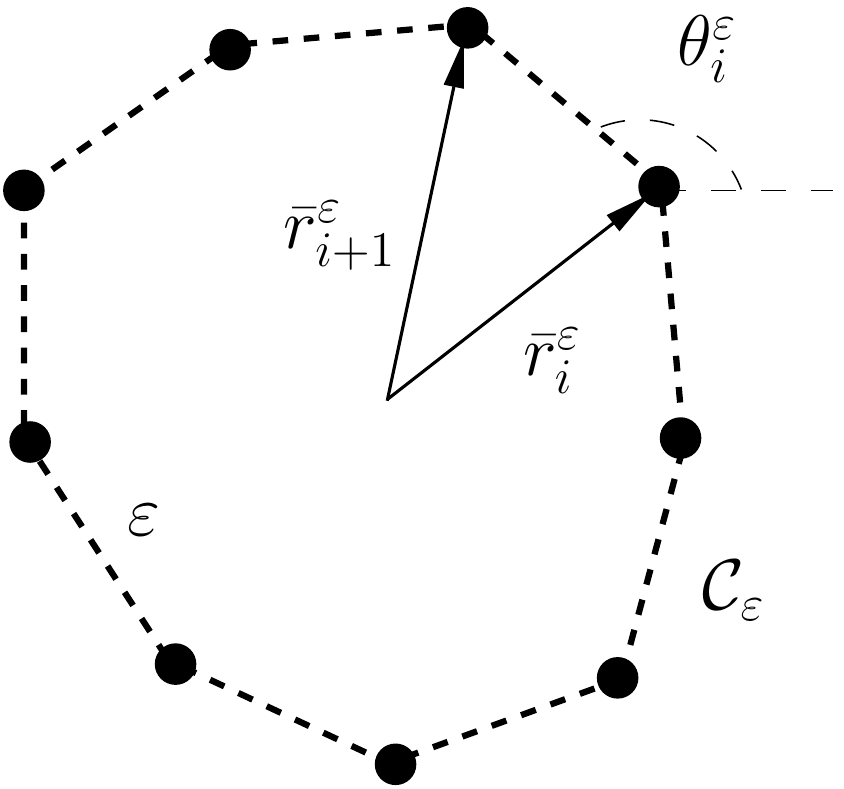}
\caption{Geometry of the problem: a discrete chain $\mathcal C^\varepsilon$ and an associated piecewise-affine curve $\mathcal C_\varepsilon$.}
  \label{fig:1}
\end{figure}
For notational convenience in what follows, we 
define $\bar r_{0}^\varepsilon
:= \bar r_{N_{\varepsilon}}^\varepsilon$ and $\bar
r_{N_\varepsilon+1}^\varepsilon := \bar r_{1}^\varepsilon$.
Given a chain $\mathcal C^\varepsilon$, we can choose
a collection of angles $\theta_i^\varepsilon$ that satisfy
\begin{equation}
\bar r_{i+1}^\varepsilon - \bar r_{i}^\varepsilon 
= 
\varepsilon(\cos(\theta_i^\varepsilon),
\sin(\theta_i^\varepsilon)),\ \ i=0,\ldots,N_\varepsilon, \label{ee46}
\end{equation}
and that satisfy
$\theta_{i}^\varepsilon-\theta_{i-1}^\varepsilon\in(-\pi,\pi)$ for
every $i=1,\ldots,N_\varepsilon$.  (We can make the choice of
$\theta_i^\varepsilon$ for $i=0,\ldots,N_\varepsilon$ unique by also
requiring that, say, $\theta_1^\varepsilon \in (-\pi, \pi]$.)
As a consequence
$\theta_0^\varepsilon=\theta^\varepsilon_{N_\varepsilon}-2k\pi$ for
some integer $k$, even though $\theta_{0}^{\varepsilon}$ and
$\theta_{N_{\varepsilon}}^{\varepsilon}$ are angles that both
correspond to the vector $\bar r_{1}^\varepsilon - \bar r_{N_{\varepsilon}}^\varepsilon$.
Based on the geometry of our problem, we shall consider only chains for which $k=1$ (if $k\geq2$ the piecewise-affine curve $\mathcal C_\varepsilon$ must self-intersect).

In what follows, we denote a vector of angles associated with a chain by
$\Theta^\varepsilon:=\left(\theta_1^\varepsilon,\ldots,\theta_{N_\varepsilon}^\varepsilon\right)\in\mathbb{R}^{N_\varepsilon}$.
Because $\sum_{i=1}^{N_\varepsilon} (\bar r_{i}^\varepsilon -
\bar r_{i-1}^\varepsilon) = 0,$ the vector $\Theta^\varepsilon$
satisfies the constraints
\begin{equation}
  \sum_{i=1}^{N_\varepsilon}\cos(\theta_{i}^\varepsilon) 
  = 
  \sum_{i=1}^{N_\varepsilon} \sin(\theta_{i}^\varepsilon) = 0. \label{ee44}
\end{equation}
Note that, for a given $\varepsilon$, there is a one-to-one correspondence between
$\Theta^\varepsilon$ and
$\mathcal C^\varepsilon$, up to a rigid rotation and translation of
$\mathcal C^\varepsilon$ in $\mathbb R^2$.

We define the energy ${\mathcal E}^{\varepsilon}[\mathcal
  C^\varepsilon]$ of a chain $\mathcal C^\varepsilon$ by
\begin{equation}
\label{eq:prop}
{\mathcal E}^{\varepsilon}[\mathcal C^\varepsilon]:=\frac{1}{\varepsilon} \sum_{i=1}^{N_{\varepsilon}} f\left( \frac{\left(\bar r_{i+1}^\varepsilon - \bar r_{i}^\varepsilon\right)\cdot\left(\bar r_{i}^\varepsilon - \bar r_{i-1}^\varepsilon\right)}{\varepsilon^2}\right),
\end{equation}
where the function $f\colon (-1,1] \to \mathbb{R}$ satisfies
\begin{equation}
\ f\in C^{\infty}((-1,1]),\ \ f'(x) < 0
  \mbox{ on}(-1,1],\ \ \lim_{x\to -1^+} f(x) = \infty,\mbox{ and }f(1)=0.
  \label{ee43}  
\end{equation}
We can rewrite the energy functional \eqref{eq:prop} in terms of angles
\begin{equation}
{\mathcal E}^{\varepsilon}[\mathcal C^\varepsilon]=E^{\varepsilon}[\Theta^\varepsilon]:=\frac{1}{\varepsilon} \sum_{i=1}^{N_{\varepsilon}} f\left( \cos(\theta_i^\varepsilon -\theta_{i-1}^\varepsilon) \right).
\end{equation}

To simplify the notation, we introduce the function $\psi\colon (-\pi,\pi)\to \mathbb{R}$ as
\[\psi(\theta) := f(\cos(\theta)).\]
By \eqref{ee43}
$\psi$ is an infinitely differentiable, even function on
$(-\pi,\pi)$ satisfying
\begin{equation}
\label{eq:prop_psi}
\psi(\theta)>\psi(0)=0\mbox{ for every }\theta\in(-\pi,0)\cup(0,\pi),\ \ \psi''(0)>0,\ \ \lim_{x\to \pm \pi} \psi(\theta) = \infty.
\end{equation}
We next set
\begin{equation}
\label{eq:1.23}
\psi_{\varepsilon}(\xi) :=\varepsilon^{-2}\psi(\varepsilon \xi)
\end{equation}
for every $\xi\in\left(-\frac{\pi}{\varepsilon},\frac{\pi}{\varepsilon}\right)$, so that
\begin{equation}
  E^{\varepsilon}[\Theta^\varepsilon]
  =
  \varepsilon \sum_{i=1}^{N_{\varepsilon}}
  \psi_{\varepsilon}\left(\frac{\theta_i^\varepsilon
    -\theta_{i-1}^\varepsilon}{\varepsilon} \right).
  \label{ee39}
\end{equation}

The {\em admissible set of angles} $T_{N_\varepsilon}$ is defined by
%
\begin{multline} 
  T_{N_\varepsilon}
  :=
  \left\{\Theta^\varepsilon\in\mathbb
    R^{N_\varepsilon}
    :
    \text{\eqref{ee44} is satisfied, }
    |\theta_{i}^\varepsilon-\theta_{i-1}^\varepsilon|<\pi
    \ \mbox{for}\
    i=2,\ldots,N_\varepsilon\vphantom{\sum_{i=1}^{N_\varepsilon}},
    \right. \\ 
    \left.
    \text{and }
    |\theta_{1}^\varepsilon-(\theta_{N}^\varepsilon-2\pi)|<\pi
  \phantom{\int\hspace{-3mm}}\right\}. \label{eq:adm} 
\end{multline}
Here for $\Theta^\varepsilon\in T_{N_\varepsilon}$, we define
$\theta_0^\varepsilon=\theta^\varepsilon_{N_\varepsilon}-2\pi$, 
which we need to compute the right-hand side of \eqref{ee39}.
We now consider the discrete minimization problem
\begin{equation}\label{eq:minprob}
\Theta_{\mathrm{min}}^\varepsilon
=
\argmin_{\Theta^\varepsilon\in
  T_{N_\varepsilon}}E^\varepsilon\left[\Theta^\varepsilon\right].
\end{equation}

Although the geometry
of our problem demands that the piecewise-affine curve $\mathcal C_\varepsilon$ associated with $\Theta^\varepsilon$ must not be self-intersecting,
we do not impose a corresponding condition on the members of
$T_{N_\varepsilon}$. Indeed, since we are interested in {\em
  minimizers} of $E^\varepsilon[\Theta^\varepsilon]$ over the
physically relevant admissible set $T^\prime_{N_\varepsilon}\subset
T_{N_\varepsilon}$, if we can show that a minimizer of
$E^\varepsilon[\Theta^\varepsilon]$ over the larger set
$T_{N_\varepsilon}$ is not self-intersecting, it is also a minimizer
of $E^\varepsilon[\Theta^\varepsilon]$ over
$T^\prime_{N_\varepsilon}$. 

The problem \eqref{eq:minprob} has a (unique) solution
that corresponds to a non self-intersecting curve in $\mathbb R^2$. Indeed, let $\tilde \Theta^\varepsilon$ satisfy $\tilde
\theta_i^\varepsilon-\tilde \theta_{i-1}^\varepsilon=
2\pi/N_\varepsilon$ for $i=1,\ldots,N_\varepsilon$. The chain corresponding to $\tilde
\Theta^\varepsilon$, which we denote by $\tilde{\mathcal
  C}^\varepsilon$, is the set of vertices for
a regular $N_\varepsilon$-sided convex polygon $\tilde{\mathcal C}_\varepsilon$ in $\mathbb
R^2$. It can
be easily verified that $\tilde \Theta^\varepsilon\in
T_{N_\varepsilon}$ and
\begin{equation}
\label{eq:polygon}
E^{\varepsilon}[\tilde \Theta^\varepsilon]= \varepsilon N_\varepsilon \psi_\varepsilon\left(\frac{2\pi}{\varepsilon N_\varepsilon}\right) = \psi_\varepsilon\left(2\pi\right).
\end{equation}
Furthermore, we have the following proposition.
\begin{proposition}
There exists an $\varepsilon_0>0$ such that $\tilde\Theta^\varepsilon$ is a minimizer of $E^\varepsilon$ over $T_{N_\varepsilon}$ when $\varepsilon<\varepsilon_0$.
\end{proposition}
\begin{proof}
By \eqref{eq:prop_psi} we can choose $\delta>0$ such that the function $\psi$ is convex on the interval $(-\delta,\delta).$ Suppose that $M\geq\psi_\varepsilon(2\pi)$ and consider an arbitrary $\Theta^\varepsilon\in T_{N_\varepsilon}$ satisfying $E^\varepsilon\left[\Theta^\varepsilon\right]\leq M$. Then 
\[\varepsilon M\geq\varepsilon^2
  \psi_{\varepsilon}\left(\frac{\theta_i^\varepsilon
    -\theta_{i-1}^\varepsilon}{\varepsilon} \right)=\psi\left(\theta_i^\varepsilon
    -\theta_{i-1}^\varepsilon\right)\geq0,\]
and by \eqref{ee43} and \eqref{eq:prop_psi} it follows that
\begin{equation}
\label{eq:1.24}
\left|\theta_i^\varepsilon
    -\theta_{i-1}^\varepsilon\right|\leq \delta,
\end{equation}
for every $i=1,\ldots,N_\varepsilon$ uniformly in $\varepsilon$ when $\varepsilon<\varepsilon_0$ and $\varepsilon_0$ is sufficiently small.
Using the discrete version of Jensen's inequality, we have 
$$E^\varepsilon[\Theta^\varepsilon]=\varepsilon\sum_{i=1}^{N_\varepsilon}
\psi_\varepsilon\left(\frac{\theta_i^\varepsilon -
  \theta_{i-1}^\varepsilon}{\varepsilon}\right)\geq \varepsilon
N_\varepsilon \psi_\varepsilon \left( \sum_{i=1}^{N_\varepsilon}
\frac{\theta_i^\varepsilon - \theta_{i-1}^\varepsilon}{\varepsilon
  N_\varepsilon}\right)=\psi_\varepsilon \left(2\pi\right) =
E^\varepsilon[\tilde \Theta^\varepsilon],$$ for any
$\Theta^\varepsilon\in T_{N_\varepsilon}$ as long as $\varepsilon<\varepsilon_0$.  We conclude that the
minimum of $E^\varepsilon$ is achieved at $\tilde
\Theta^\varepsilon$.
\end{proof}

\section{Continuum Formulation} \label{s2}

To motivate the subsequent developments, observe that \eqref{eq:prop_psi}, \eqref{eq:polygon}, and the smallness of $\varepsilon$ give
\[E^{\varepsilon}[\tilde \Theta^\varepsilon]=  \psi_\varepsilon\left(2\pi\right) = \frac{1}{\varepsilon^2} \psi(2\pi\varepsilon)=2\pi^2\psi^{\prime\prime}(0) + o(1).\]
It follows that
\begin{equation}
\label{eq:elim_poly}
\lim_{\varepsilon\to0}\left\{\inf_{\Theta^\varepsilon\in T_{N_\varepsilon}} E^\varepsilon [\Theta^\varepsilon]\right\}=\lim_{\varepsilon\to0}E^\varepsilon [\tilde\Theta^\varepsilon]=2\pi^2\psi^{\prime\prime}(0)
\end{equation}
and therefore
\begin{equation}
\inf_{\Theta^\varepsilon\in T_{N_\varepsilon}} E^\varepsilon[\Theta^\varepsilon]<C,
\end{equation}
for some $C>0$ uniformly in $\varepsilon$ when $\varepsilon$ is small
enough.  As noted in the previous
section, the curve $\tilde{\mathcal C}_\varepsilon$ associated with the
minimizer $\tilde\Theta^\varepsilon$ is a regular
$N_\varepsilon$-sided convex polygon in $\mathbb{R}^2$.  When
$\varepsilon\to0$, the number of sides
$N_\varepsilon=1/\varepsilon\to\infty$ while the total perimeter of
the polygon remains equal to $\varepsilon N_\varepsilon=1.$ The
sequence $\left\{\tilde{\mathcal
  C}_\varepsilon\right\}_{\varepsilon>0}$ thus converges (uniformly) to a circle
$\mathcal C_0$ of radius $\frac{1}{2\pi}$ when $\varepsilon\to0$. 
Based on \eqref{eq:elim_poly}, 
it seems natural to associate to this limiting circle $\mathcal C_0$
the energy $E_0[\mathcal C_0]:=2\pi^2\psi^{\prime\prime}(0)$.

Further, given an arbitrary smooth, simple curve $\mathcal
C\subset\mathbb R^2$ such that there exists a sequence of 
piecewise-affine curves
$\mathcal C_\varepsilon$ (which corresponds to a sequence of chains
$\mathcal C^\varepsilon$) converging to $\mathcal C ,$ it might be
tempting to extend the notion of energy to $\mathcal C$ by defining
$E_0[\mathcal C]:=\lim_{\varepsilon\to0} E_\varepsilon[\mathcal
  C_\varepsilon]$. However, a priori it is not clear that this limit
exists or if its value is the same for all sequences of chains
converging to $\mathcal C$. In addition, if the notion of the limiting
energy for curves can be made precise, it would be desirable that
minimizers of the discrete problem
$\Theta_{\mathrm{min}}^\varepsilon=\argmin_{\Theta^\varepsilon\in
  T_{N_\varepsilon}}E^\varepsilon\left[\Theta^\varepsilon\right]$ for
chains converge to minimizers of the limiting energy $E_0$ over an
appropriate function space. The established framework to study
convergence of energies that preserves the variational structure of
the discrete problem is that of $\Gamma$-convergence, which we consider next.


Before proving $\Gamma$-convergence, we need to select a common function space that contains both the discrete chains $C^\varepsilon$ and the limiting curves. Note first that, since the limiting energy should correspond to Euler's elastica, the curvature of a limiting curve $\mathcal C$ must be square integrable, i.e., the angle function for $\mathcal C$ must be an element of the Sobolev space $H^1([0,1])$. Our goal is to show that the Euler's elastica energy of a closed curve with an angle function in $H^1([0,1])$ is the limit of the discrete energies of a sequence of chains.
However, the angle function for a chain is a step function and hence
is not in $H^1([0,1])$. 
To put our construction into the framework of $\Gamma$-convergence, we use the idea of \cite{braides2002gamma} and replace a sequence of piecewise-constant angle functions for chains by a sequence of piecewise-affine functions in $H^1([0,1])$. We then introduce an energy functional defined over the piecewise-affine functions so that for each $\varepsilon$, the new energy of each piecewise-affine function is the same as the old discrete energy of a corresponding chain.  
This yields a sequence of affine functions in $H^1([0,1])$.
Note that the affine functions considered below do not need to correspond to a closed curve on the plane and are not required to have length one.  The physically relevant geometric constraints are imposed only on the piecewise-constant angle functions for the discrete chains and on the limiting angle function.

To make these ideas more precise, consider a partition of the interval $[0,1]$ by the
points \[\left\{0, \varepsilon/2, 3\varepsilon/2, \ldots,
1-\varepsilon/2,1\right\}\] and denote by
$\tilde A_{\varepsilon}(0,1)\subset C([0,1])$ the set of functions
affine on each subinterval of this partition. From
now on, we will identify with a vector $\Theta^\varepsilon\in
T_{N_\varepsilon}$ a piecewise-affine function $\theta_\varepsilon\in
\tilde A_{\varepsilon}(0,1)$ given by
\begin{multline}
\label{eq:theta_fun}
\theta_\varepsilon(s):=\left(\frac{\theta^\varepsilon_1-2\pi+\theta_{N_\varepsilon}^\varepsilon}{2}+\frac{s}{\varepsilon}\left(\theta^\varepsilon_1+2\pi-\theta_{N_\varepsilon}^\varepsilon\right)\right)\chi_{\left[0,\frac{\varepsilon}{2}\right)}(s) \\
+\sum_{i=1}^{N_\varepsilon-1}\left(\theta_i^\varepsilon+\frac{\theta_{i+1}^\varepsilon-\theta_i^\varepsilon}{\varepsilon}\left(s-\frac{2i-1}{2}\varepsilon \right)\right)\chi_{\left[\frac{2i-1}{2}\varepsilon,\frac{2i+1}{2}\varepsilon\right)}(s) \\
+\left(\frac{\theta_{1}^\varepsilon+2\pi+\theta_{N_\varepsilon}^\varepsilon}{2}+\frac{s-1}{\varepsilon}\left(\theta_{1}^\varepsilon+2\pi-\theta_{N_\varepsilon}^\varepsilon\right)\right)\chi_{\left[1-\frac{\varepsilon}{2},1\right]}(s)
\end{multline}
for $s\in[0,1]$, where $\chi_S$ is the indicator function of the set $S\subset\mathbb R$. Then
\begin{equation}
\label{eq:Thet_thet}
\theta_\varepsilon\left(\frac{2i-1}{2}\varepsilon\right)=\theta^\varepsilon_i
\end{equation}
for $i=1,\ldots,N_\varepsilon$ and
\begin{equation}
\label{eq:degree}
  \theta_\varepsilon(0)
  =
  \frac{\theta^\varepsilon_1+\theta_{N_\varepsilon}^\varepsilon}{2}-\pi,
  \quad 
  \theta_\varepsilon(1)=\frac{\theta_{1}^\varepsilon+\theta_{N_\varepsilon}^\varepsilon}{2}+\pi,    
\end{equation}
so that $\theta_\varepsilon(1)=\theta_\varepsilon(0)+2\pi$.

We use \eqref{eq:Thet_thet}--\eqref{eq:degree} and the definition of
$T_{N_\varepsilon}$ to define the {\em admissible set of functions}
\begin{multline}
\label{eq:adm_fun}
A_{\varepsilon}(0,1):=\left\{\theta\in \tilde A_{\varepsilon}(0,1)\colon \theta(0)=\frac{\theta(\varepsilon/2)+\theta(1-\varepsilon/2)}{2}-\pi,\right.\\\left.\theta(1)=\frac{\theta(\varepsilon/2)+\theta(1-\varepsilon/2)}{2}+\pi;\left(\theta\left(\varepsilon/2\right),\ldots,\theta\left(1-\varepsilon/2\right)\right)\in T_{N_\varepsilon}\right\}.
\end{multline}
Note that for $\theta_\varepsilon\in A_{\varepsilon}(0,1)$, we have
\begin{equation}
\label{eq:theta_fun_der}
\theta^\prime_\varepsilon(s):=\frac{\theta_{1}^\varepsilon+2\pi-\theta_{N_\varepsilon}^\varepsilon}{\varepsilon}\left(\chi_{\left[0,\frac{\varepsilon}{2}\right)}(s)+\chi_{\left(1-\frac{\varepsilon}{2},1\right]}(s)\right)
+\sum_{i=1}^{N_\varepsilon-1}\frac{\theta_{i+1}^\varepsilon-\theta_i^\varepsilon}{\varepsilon}\chi_{\left(\frac{2i-1}{2}\varepsilon,\frac{2i+1}{2}\varepsilon\right)}(s)
\end{equation}
for all $s\in[0,1]$, where
$\theta_i^\varepsilon=\theta_{\varepsilon}\left((2i-1)\varepsilon/2\right)$
for $i=1,\ldots,N_{\varepsilon}$.
It follows from \eqref{eq:theta_fun_der} that
\begin{equation}
  \int_0^1\psi_\varepsilon\left(\theta_\varepsilon^\prime\right)\,ds
  =
  \varepsilon\sum_{i=1}^{N_\varepsilon}
  \psi_\varepsilon\left(\frac{\theta^\varepsilon_{i}-\theta_{i-1}^\varepsilon}{\varepsilon}\right)
  =
  E^\varepsilon\left[\Theta^\varepsilon\right],  \label{eq:int}
\end{equation}
for every $\theta_\varepsilon\in A_\varepsilon(0,1)$,
where $\theta_0^\varepsilon=\theta_{N_{\varepsilon}}^\varepsilon-2\pi$.

If we define the functional $F_\varepsilon:H^1\left([0,1];\mathbb{R}\right)\to \bar{\mathbb{R}}$ by
\begin{equation}
\label{eq:ffun}
F_\varepsilon[\theta]:=\left\{
\begin{array}{ll}
\int_0^1\psi_\varepsilon\left(\theta^\prime\right)\,ds, & \theta\in A_\varepsilon(0,1), \\
\infty, & \mathrm{otherwise},
\end{array}
\right.
\end{equation}
for every $\varepsilon>0$, then \eqref{eq:int} implies that
\begin{equation}
\label{eq:equiv}
F_\varepsilon\left[\theta_\varepsilon\right]=E^\varepsilon\left[\Theta^\varepsilon\right],
\end{equation}
whenever $\theta_\varepsilon\in A_\varepsilon(0,1)$, where
$\Theta^\varepsilon\in\mathbb{R}^{N_\varepsilon}$ is the vector
corresponding to $\theta_\varepsilon$.  The discrete minimization
problem \eqref{eq:minprob} has an associated continuum minimization
problem
\begin{equation}
  \label{eq:contminprob}
  \theta_{\varepsilon,\min}
  =
  \argmin_{\theta\in H^1([0,1])}F_\varepsilon\left[\theta\right].
\end{equation}
Because the functionals $\{ F_\varepsilon \}_{\varepsilon>0}$ in
\eqref{eq:contminprob} are all defined on the same space $H^1([0,1])$,
an asymptotic limit of $\{ F_\varepsilon \}_{\varepsilon>0}$ can be studied using
$\Gamma$-convergence.
%

\section{The $\Gamma$-Limit} \label{s3}

In this section we state and prove the asymptotic limit of the
sequence of continuum energies $\{ F_\varepsilon \}_{\varepsilon>0}$. We state two lemmas, whose proofs are in Appendices 1 and 2. The first lemma shows that the constraints imposed on the piecewise-constant angle functions are preserved under the weak-$H^1$ convergence.

\begin{lemma} \label{r1}
Suppose $\{\theta_\varepsilon\}_{\varepsilon>0}$ converges
weakly to $\theta$ in $H^1([0,1])$ as $\varepsilon\rightarrow 0$.  If
there is a sequence $\{ \varepsilon_{n}  \}$ of positive numbers such that 
$\underset{n\rightarrow \infty}{\lim }\varepsilon_{n}=0$ and
$\theta_{\varepsilon_{n}}\in A_{\varepsilon_{n}}(0,1)$ for all $n$, 
then 
\begin{equation}
  \int_0^1\cos(\theta)\,ds
  =
  \int_0^1\sin(\theta)\,ds
  =0
  \quad
  \text{and}
  \quad
  \theta(1)-\theta(0)=2\pi.  \label{ee45}
\end{equation}
\end{lemma}

The second lemma establishes that any function in $H^1([0,1])$ satisfying \eqref{ee45} can be approximated by a twice continuously differentiable function on $[0,1]$ that also satisfies \eqref{ee45}.
\begin{lemma} \label{r2}
Suppose $\theta\in H^1([0,1])$ satisfies \eqref{ee45}.  Then for all
$\delta>0$ there is a function $\theta^{*}\in C^{2}([0,1])$ such
that $\|\theta-\theta^{*}\|_{H^{1}([0,1])}<\delta$ and 
$\theta^{*}$ also satisfies \eqref{ee45}.
\end{lemma}

Now we state the main result of this paper.
\begin{theorem}[$\Gamma$-convergence]
Let $F_\varepsilon \colon$$H^1([0,1])\to\bar{\mathbb{R}}$ be
defined by \eqref{eq:ffun}.
Let $E_0 \colon$$H^1([0,1])\to\bar{\mathbb{R}}$ be defined by
\begin{equation}
\label{eq:ffun0}
  E_0(\theta)
  :=
  \left\{
  \begin{array}{ll}
    \int_0^1 \alpha(\theta')^2\,ds, & \theta\in H^1_c([0,1]), \\
    \infty, & \mathrm{otherwise},
  \end{array}
  \right.
\end{equation}
where $\alpha = \psi''(0)/2$ and
\begin{equation}
  H^1_c([0,1])
  :=
  \left\{\theta\in H^1([0,1])\ :\ \theta \ \text{satisfies}\ \eqref{ee45}
  \right\}.
  \label{eq:cont_adm_fun}
\end{equation}
%
%
%
Then $\Gamma\text{-}\lim_{\varepsilon\to 0} F_{\varepsilon} = E_0$ in the weak topology of $H^1([0,1])$, that is
\smallskip
\begin{enumerate}[label=(\alph*)]
\item For every $\theta \in H^1([0,1])$, there exists a
sequence $\{\theta_\varepsilon\}_{\varepsilon>0}$ converging weakly to $\theta$ in
$H^1([0,1])$ such that $\lim_{\varepsilon\to 0}
F_\varepsilon[\theta_\varepsilon] = E_0[\theta]$. \smallskip
\item For every
sequence $\{\theta_\varepsilon\}_{\varepsilon>0}$ converging weakly to $\theta$ in
$H^1([0,1])$, \[\liminf_{\varepsilon\to 0}
F_\varepsilon[\theta_\varepsilon] \geq E_0[\theta].\]
\end{enumerate}
\smallskip
Furthermore, if a sequence $\left\{\theta_\varepsilon\right\}_{\varepsilon>0}\subset H^1([0,1])$ satisfies a uniform energy bound $F_\varepsilon\left[\theta_\varepsilon\right]<C$ then there is a subsequence $\{\theta_{\varepsilon_j}\}$ such that $\theta_{\varepsilon_j}\stackrel{H^1}\rightharpoonup\theta$ as $\varepsilon_j\to 0$ for some $\theta\in H^1([0,1]).$
\end{theorem}

Note that the last assertion of the theorem also tells us that, if there is a sequence of chains $\left\{\Theta^\varepsilon\right\}_{\varepsilon>0}$ that satisfies a uniform energy bound $E^\varepsilon\left[\Theta^\varepsilon\right]<C$, then there is a subsequence of the corresponding affine angle functions $\{\theta_{\varepsilon_j}\}$ such that $\theta_{\varepsilon_j}\stackrel{H^1}\rightharpoonup\theta$ as $\varepsilon_j\to 0$ for some $\theta\in H^1([0,1]).$ One can easily check that the sequence of piecewise-constant angle functions with values given by $\left\{\Theta^\varepsilon\right\}_{\varepsilon>0}$ converges to the same $\theta$ strongly in $L^2([0,1])$.

\begin{proof} We begin by proving the final statement of the theorem. In what follows $C$ denotes a generic positive constant. Suppose that a sequence $\left\{\theta_\varepsilon\right\}_{\varepsilon>0}\subset H^1([0,1])$ satisfies a uniform energy bound $F_\varepsilon\left[\theta_\varepsilon\right]<C$. Then \eqref{eq:ffun} implies that $\theta_\varepsilon\in A_\varepsilon (0,1)$ so that $\theta_\varepsilon$ is piecewise-affine for every $\varepsilon$. By \eqref{eq:int} and the definition \eqref{eq:1.23} of $\psi_\varepsilon$, we have 
\[\frac{1}{\varepsilon}\sum_{i=1}^{N_\varepsilon}
  \psi\left({\theta^\varepsilon_{i}-\theta_{i-1}^\varepsilon}\right)=\varepsilon\sum_{i=1}^{N_\varepsilon}
  \psi_\varepsilon\left(\frac{\theta^\varepsilon_{i}-\theta_{i-1}^\varepsilon}{\varepsilon}\right)\leq C.\]
Because $\psi(0)=0$ is the unique global minimum of $\psi$, the previous estimate implies
\[\psi\left({\theta^\varepsilon_{i}-\theta_{i-1}^\varepsilon}\right)\leq C\varepsilon\]
for all $\varepsilon$. By a similar argument that led to \eqref{eq:1.24}, we obtain 
\[
\left|\theta_i^\varepsilon
    -\theta_{i-1}^\varepsilon\right|\leq \delta,
\]
for every $i=1,\ldots,N_\varepsilon$ where $\delta=o(1)$ in $\varepsilon$. This along with \eqref{eq:prop_psi} enables us to conclude that 
\[{\left({\theta^\varepsilon_{i}-\theta_{i-1}^\varepsilon}\right)^2}\leq C\psi\left({\theta^\varepsilon_{i}-\theta_{i-1}^\varepsilon}\right)\] 
for all $i=1,\ldots,N_\varepsilon$ and some $C>0$ when $\varepsilon$ is sufficiently small. This yields the inequality
\[ \int_0^1\left(\theta_\varepsilon^\prime\right)^2\,ds=\varepsilon\sum_{i=1}^{N_\varepsilon}
  \left(\frac{\theta^\varepsilon_{i}-\theta_{i-1}^\varepsilon}{\varepsilon}\right)^2\leq\varepsilon C\sum_{i=1}^{N_\varepsilon}
  \psi_\varepsilon\left(\frac{\theta^\varepsilon_{i}-\theta_{i-1}^\varepsilon}{\varepsilon}\right)\leq C.\]
Because we can always assume that $\theta_\varepsilon(0)\in[-\pi,\pi]$, the boundedness and hence the weak compactness of $\left\{\theta_\varepsilon\right\}_{\varepsilon>0}$ in $H^1([0,1])$ now follow from the Poincare inequality in one dimension.

\smallskip
We now proceed with proving $\Gamma$-convergence.

\smallskip
{\bf Proof of (a): Construction of the recovery sequence.}  Let $\theta \in H^1([0,1])$.  Suppose $\theta \notin
H^1_c([0,1])$, so that $E_0[\theta] = \infty$.  We define a constant
sequence by setting $\theta_\varepsilon = \theta$ for all $\varepsilon$.  If
$\theta_{\varepsilon}\in A_\varepsilon (0,1)$ for arbitrarily small
$\varepsilon$, then Lemma~\ref{r1}
would imply that $\theta \in H^1_c([0,1])$.  So there is an
$\bar{\varepsilon}>0$ such that $\theta=\theta_\varepsilon \notin
A_{\varepsilon}(0,1)$ and hence
$F_\varepsilon\left[\theta_\varepsilon\right]=\infty$ for all $0<\varepsilon\leq
\bar{\varepsilon}$.

We assume now that $\theta \in H^1_c([0,1])$.  By Lemma~\ref{r2}, we
can assume as well that $\theta\in C^{2}([0,1])$.  Working with a
smooth function will allow us to bound first and second derivatives of
$\theta$ uniformly on $[0,1]$, which we need to do for several later
estimates.  
Recall that in what follows $\varepsilon>0$ is such
that $N_\varepsilon=1/\varepsilon$ is in $\mathbb{N}$.
We divide the rest of the proof of (a) into several steps.

\smallskip
{\bf Step 1:}  Let $\bar{r}$ denote the curve whose angle function is
$\theta$.  We construct a chain with $N_\varepsilon$ sides 
that is uniformly close to $\bar{r}$. Later we shall demonstrate that the corresponding affine
function---which has the same energy as the discrete energy of the chain---approximates $\theta$ in $H^1([0,1])$.

\smallskip
Since $\bar{r}$ and any admissible chains have length $1$, we cannot inscribe an
admissible chain in
$\bar{r}$.  Instead, for
$h>0$ we define the `inflated' curve
$\bar{r}_{h}(s)=\bar{r}(s)+h\bar{N}(s)$, where $\bar{N}$ denotes the
(outward) normal to the curve $\bar{r}$ (see Fig. \ref{fig:2}).  
\begin{figure}[htb]
\centering
\includegraphics[height=1.5in]{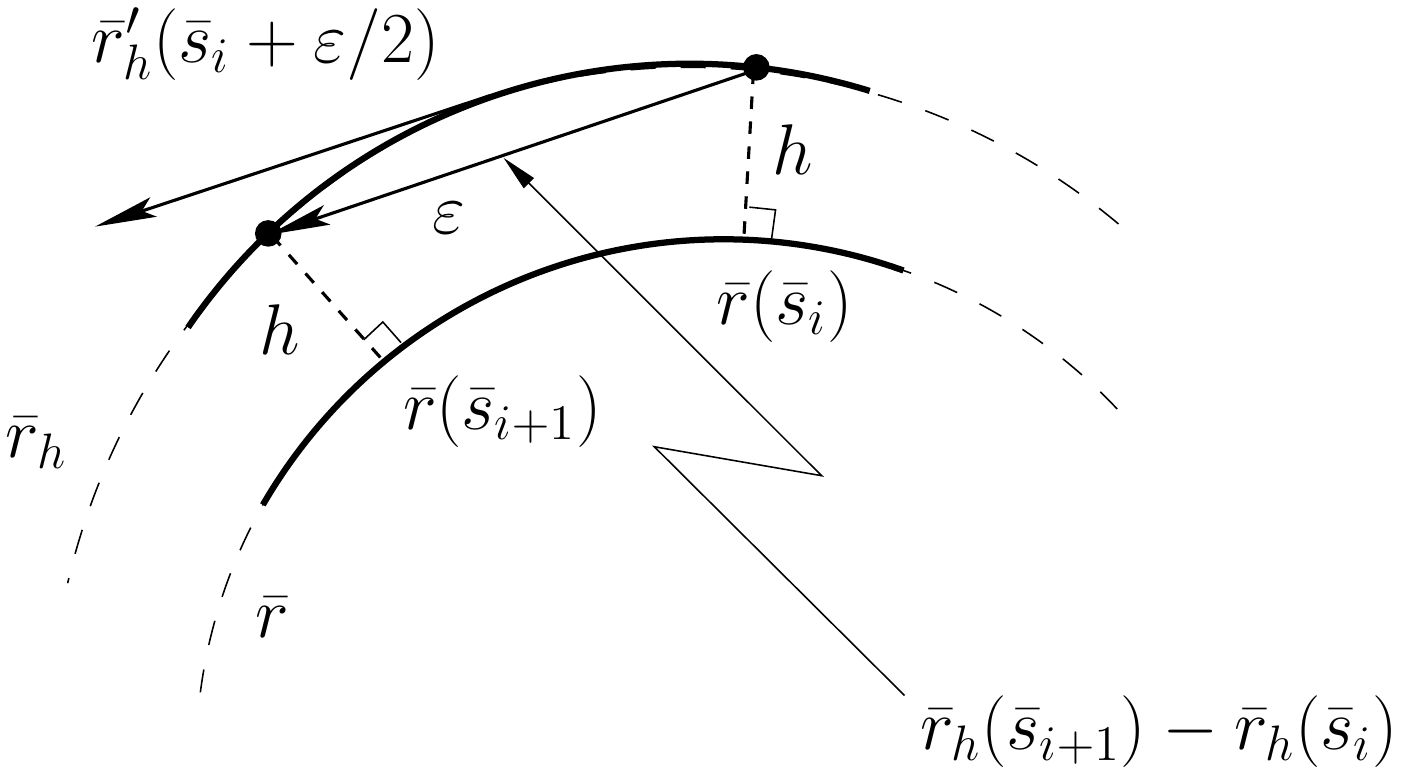}
\caption{A secant line and tangent vector for the inflated curve $\bar{r}_{h}$.}
  \label{fig:2}
\end{figure}
The length of $\bar{r}_{h}$
is $1+2\pi h$.  
Given $\varepsilon$ sufficiently small,
it is clear that there
exists an $h$ such that we can inscribe
a chain with $N_\varepsilon$ 
sides each of length $\varepsilon$ in $\bar{r}_{h}$.  So there
exist $\bar{s}_{1},\ldots,\bar{s}_{N_\varepsilon}\in [0,1]$ such that
$\left\{\bar{r}_h(\bar{s}_{i})\right\}_{i=1}^{N_\varepsilon}$ is a
chain.  Without loss of generality we can assume that $\bar{s}_{1}=0$.
%
%

We let  $\bar{\Theta}^\varepsilon:=\left(\bar{\theta}_1,\ldots,\bar{\theta}_{N_\varepsilon}\right)\in\mathbb{R}^{N_\varepsilon}$
denote the vector of angles associated with the chain (see \eqref{ee46}).
Then $\bar{\Theta}^{\varepsilon}\in T_{N_\varepsilon}$ and has a discrete energy
$E^{\varepsilon}[\bar{\Theta}^{\varepsilon}]$ defined by \eqref{ee39}.
By \eqref{eq:theta_fun} we construct  a piecewise-affine
function $\hat{\theta}_{\varepsilon}$ such that
$F_\varepsilon[\hat{\theta}_{\varepsilon}]=E^{\varepsilon}[\bar{\Theta}^{\varepsilon}]$. 
Our goal is to show that 
$\hat{\theta}_{\varepsilon}$ is close to $\theta$ in $H^1([0,1])$ and that
$F_\varepsilon[\hat{\theta}_{\varepsilon}]$ is close to $E_0[\theta]$.

\smallskip
{\bf Step 2:}  We derive two preliminary estimates 
\begin{equation}
 \bar{s}_{i+1}
 =
 \bar{s}_{i} + \varepsilon + O(\varepsilon^{3}) 
 \quad \text{and} \quad
 \bar{\theta}_{i}
 =
 \theta(\bar{s}_{i}+\varepsilon/2) + O(\varepsilon^{2}).
 \label{ee47}
\end{equation}
We begin with an initial estimate on
$\bar{s}_{i+1}-\bar{s}_{i}$.  To attain this, we define 
$F(\sigma):=|\bar{r}_{h}(\sigma)-\bar{r}_{h}(\bar{s}_{i})|$. One
can check that $F_{\sigma}(\bar{s}_{i})=1+h\theta'(\bar{s}_{i})$ and that
$F_{\sigma\sigma}(\bar{s}_{i})=h\theta''(\bar{s}_{i})$.  Because
$F_{\sigma}(\bar{s}_{i})\neq 0$ for
$h$ sufficiently small, the equation $F(\sigma)=\varepsilon$
defines a function  ${\sigma}(\varepsilon)$ for 
small $\varepsilon$, where $\sigma(\varepsilon)=\bar{s}_{i+1}$.  Observe that $F({\sigma}(\varepsilon))=\varepsilon$ implies
\begin{equation}
  F_{\sigma}\sigma_{\varepsilon}=1,
  \qquad
  F_{\sigma\sigma}\sigma^{2}_{\varepsilon}
  +
  F_{\sigma}\sigma_{\varepsilon\varepsilon}
  = 0,
  \label{ee24}
\end{equation}
so that
\begin{align}
  \sigma_{\varepsilon}
  &=
  (1+h\theta')^{-1}
  =1-h\theta'+O(h^{2}), \label{ee25}\\
  \sigma_{\varepsilon\varepsilon}
  &=
  -F_{\sigma\sigma}\sigma^{2}_{\varepsilon}/F_{\sigma}
  =
  -\theta''h(1+h\theta')^{-3}
  =
  -\theta''h + O(h^{2}). \label{ee27}
\end{align}
Hence
\begin{equation}
  {\sigma}(\varepsilon)
  =
  {\sigma}(0)
  +
  (1-h\theta'+O(h^{2}))\varepsilon
  +
  \frac{1}{2}(-\theta''h + O(h^{2}))\varepsilon^{2}
  +
  O(\varepsilon^{3}).
  \label{ee26}
\end{equation}
Because ${\sigma}(0)=\bar{s}_{i}$ and
${\sigma}(\varepsilon)=\bar{s}_{i+1}$, \eqref{ee26}
implies that
\begin{equation}
  \bar{s}_{i+1} - \bar{s}_{i}
  =
  \varepsilon + O(h\varepsilon) + O(\varepsilon^3), \label{ee28}
\end{equation}
Building on \eqref{ee28}, we have
\begin{equation}
  1
  =
  \sum_{i=1}^{N_\varepsilon}(\bar{s}_{i+1}-\bar{s}_{i})
  =
  \sum_{i=1}^{N_\varepsilon}(\varepsilon + O(h\varepsilon) + O(\varepsilon^3))
  =
  N_\varepsilon\varepsilon + \sum_{i=1}^{N_\varepsilon}(O(h\varepsilon) + O(\varepsilon^3)).
  \label{ee29}
\end{equation}
%
%
Because $N_\varepsilon\varepsilon=1$ and $\sum_{i=1}^{N_\varepsilon}(O(h\varepsilon)+ O(\varepsilon^3))=O(h) +
O(\varepsilon^2)$, \eqref{ee29} implies that 
$0=O(h)+ O(\varepsilon^2)$, or $h=O(\varepsilon^{2})$.
%
Now, returning to \eqref{ee28} and using $h=O(\varepsilon^{2})$ implies $\subrefb{ee47}{1}$.

Next we show $\subrefb{ee47}{2}$.
Note that $\bar{\theta}_{i}-\theta(\bar{s}_{i}+\varepsilon/2)$ is the
angle between
the vectors
\mbox{$\varepsilon^{-1}( \bar{r}_{h}(\bar{s}_{i+1})-\bar{r}_{h}(\bar{s}_{i}) )$} and
$\bar{r}'_{h}(\bar{s}_{i}+\varepsilon/2)$ (see Fig.~\ref{fig:2}).
We can write
\begin{align}
  \frac{\bar{r}_{h}(\bar{s}_{i+1})-\bar{r}_{h}(\bar{s}_{i})}{\varepsilon}
  &=
  \frac{\bar{r}'_{h}(\bar{s}_{i})(\bar{s}_{i+1}-\bar{s}_{i})
         +\bar{r}''_{h}(\bar{s}_{i})(\bar{s}_{i+1}-\bar{s}_{i})^{2}/2
         +O((\bar{s}_{i+1}-\bar{s}_{i})^{3})}
       {\varepsilon} \nonumber  \\
  &=\bar{r}'_{h}(\bar{s}_{i})+\bar{r}''_{h}(\bar{s}_{i})\varepsilon/2
         +O(\varepsilon^{2}), \label{ee19}
\end{align}
where we have used $\subrefb{ee47}{1}$.  Likewise, we can write
\begin{align}
  \bar{r}'_{h}(\bar{s}_{i}+\varepsilon/2)
  &=
  \bar{r}'_{h}(\bar{s}_{i})
  +
  \bar{r}''_{h}(\bar{s}_{i})\varepsilon/2
  +
  O(\varepsilon^{2}).  \label{ee20}
\end{align}
From \eqref{ee19} and \eqref{ee20}, we see that
\begin{equation}
  \bar{r}'_{h}(\bar{s}_{i}+\varepsilon/2)
  =
  \frac{\bar{r}_{h}(\bar{s}_{i+1})-\bar{r}_{h}(\bar{s}_{i})}{\varepsilon}
  +
  O(\varepsilon^{2}),
  \label{ee21}
\end{equation}
where the leading order term on the right hand side is a unit vector. The largest angle between $\bar{r}'_{h}(\bar{s}_{i}+\varepsilon/2)$ and ${\varepsilon}^{-1}\left({\bar{r}_{h}(\bar{s}_{i+1})-\bar{r}_{h}(\bar{s}_{i})}\right)$ for a small fixed magnitude of their difference is achieved when this difference is perpendicular to ${\varepsilon}^{-1}\left({\bar{r}_{h}(\bar{s}_{i+1})-\bar{r}_{h}(\bar{s}_{i})}\right)$. It then immediately follows that
$\bar{\theta}_{i}-\theta(\bar{s}_{i}+\varepsilon/2) = O(\varepsilon^{2})$.

\smallskip
{\bf Step 3:}  We now use the estimates \eqref{ee47} to show that (i) the piecewise-affine function 
$\hat{\theta}_{\varepsilon}$ constructed at the end of Step 1 is close to $\theta$ in $H^1([0,1])$ and (ii)
the energy $F_\varepsilon[\hat{\theta}_{\varepsilon}]$ is close to $E_0[\theta]$.

(i) First we demonstrate that $\hat{\theta}_{\varepsilon}$ is close to $\theta$ in $H^1([0,1])$.
We have
\begin{align}
	\int_{0}^{1}\!\left|\hat{\theta}_{\varepsilon}'(s)-{\theta}'(s)\right|^{2}ds
        &=
        \int_{0}^{s_{1}}\!
        \left|\frac{\bar{\theta}_{1}-(\bar{\theta}_{N_\varepsilon}-2\pi+\bar{\theta}_{1})/2}{\varepsilon/2}
        -{\theta}'(s)\right|^{2}ds \nonumber \\
        &\qquad +
        \sum_{i=1}^{N_\varepsilon-1}\int_{s_{i}}^{s_{i+1}}\!\left|
        \frac{\bar{\theta}_{i+1}-\bar{\theta}_{i}}{\varepsilon}
        -{\theta}'(s)\right|^{2}ds \label{ee9}\\
        &\qquad\qquad +
        \int_{s_{N_\varepsilon}}^{1}\!
        \left|\frac{(\bar{\theta}_{1}+2\pi+\bar{\theta}_{N_\varepsilon})/2 - \bar{\theta}_{N_\varepsilon}}{\varepsilon/2}
        -{\theta}'(s)\right|^{2}ds. \nonumber
\end{align}

To estimate the right-hand side of \eqref{ee9},  recall that $\bar{r}_h(\bar{s}_i),\ i=1,\ldots,N_\varepsilon$, denote the vertices of the chain inscribed into the inflated curve $\bar{r}_h$. We observe that
\begin{equation}
	\frac{\bar{\theta}_{i+1}-\bar{\theta}_{i}}{\varepsilon}
        =
	\frac{\theta(\bar{s}_{i+1}+\varepsilon/2)-\theta(\bar{s}_{i}+\varepsilon/2)+O(\varepsilon^{2})}{\varepsilon}
        =
        \theta'(\xi_{i})+O(\varepsilon),
	\label{ee4}
\end{equation}
where the first equality uses $\subrefb{ee47}{2}$ and where 
$\xi_{i}\in (\bar{s}_{i}+\varepsilon/2,\bar{s}_{i+1}+\varepsilon/2)$.
The equations 
$\bar{s}_{1}=0$ and $\subrefb{ee47}{1}$
imply
that $\bar{s}_{i+1}=i\varepsilon + O(\varepsilon^{2})$ for
$i=1,\ldots,N_{\varepsilon}$.  Then, because
$s_{i}=(2i-1)\varepsilon/2$, one has
\begin{equation}
  |\xi-s| < 2\varepsilon + O(\varepsilon^{2})
  \quad
  \text{for $s_{i}<s<s_{i+1}$
    and
    $\bar{s}_{i}+\varepsilon/2 < \xi < \bar{s}_{i+1}+\varepsilon/2$}.
  \label{ee7}
\end{equation} 
Combining \eqref{ee7} and \eqref{ee4} yields
\begin{eqnarray}
	\int_{s_{i}}^{s_{i+1}}\!\left|
        \frac{\bar{\theta}_{i+1}-\bar{\theta}_{i}}{\varepsilon}
        -\theta'(s)\right|^{2}ds
        &=
        \int_{s_{i}}^{s_{i+1}}\!\left|\theta'(\xi_{i})-\theta'(s)+O(\varepsilon)\right|^{2}ds
	\nonumber  \\
        &=
        \int_{s_{i}}^{s_{i+1}}\!\left|\theta''(\hat{\xi}_{i})(\xi_{i}-s)+O(\varepsilon)\right|^{2}ds
	\label{ee6}\\
        &\leq
        \int_{s_{i}}^{s_{i+1}}\!O(\varepsilon^{2})\,ds=O(\varepsilon^{3}).
	\nonumber
\end{eqnarray}

We need an estimate like \eqref{ee6} for the first and third terms on
the right-hand side of \eqref{ee9}.  However, both the curve $\bar{r}$ and the associated chain constructed in Step 1 are closed in the plane, hence their parametrizations can be extended periodically with period $1$ to $\mathbb{R}$. Selecting a different vertex in the chain to correspond to $s=0$ is equivalent to translating the parametrization by a number less than $1$. In this case, the first and the third integrals in \eqref{ee9} become one of the integrals in the sum in the middle term. Thus the first and the third integrals in \eqref{ee9} together admit the same $O(\varepsilon^3)$-estimate as in \eqref{ee6}.

Returning to \eqref{ee9}, we conclude that
\begin{equation}
	\int_{0}^{1}\!\left|\hat{\theta}_{\varepsilon}'(s)-\theta'(s)\right|^{2}ds
        =
\sum_{i=1}^{N_\varepsilon}\!O(\varepsilon^{3})=O(\varepsilon^2). \label{ee13}
\end{equation}

(ii) Now we estimate the difference between
$F_{\varepsilon}[\hat{\theta}_{\varepsilon}]$ and $E_{0}[\theta]$. Straightforward but tedious calculations based upon expanding $\psi$ show that
\begin{align}
  |F_{\varepsilon}[\hat{\theta}_{\varepsilon}]-E_{0}[\theta]|
  &=
 \left\lvert
  \int_{0}^{s_{1}}\!
   \left\{\frac{\psi''(0)}{2}
  \left[
  \left(\frac{\theta'(\xi_{0}) + \theta'(\xi_{N_\varepsilon})}{2}\right)^{2}
  -
  \theta'(s)^{2}
   \right]
   + O(\varepsilon)\right\} \,ds \right.  \nonumber \\
  &\quad+
  \sum_{i=1}^{N_\varepsilon-1}
  \int_{s_{i}}^{s_{i+1}}\!
   \left\{\frac{\psi''(0)}{2}
  \left[
  \theta'(\xi_{i})^{2}
  -
  \theta'(s)^{2}
  \right]
   + O(\varepsilon) \right\}\,ds
  \nonumber \\
  &\quad+
    \left.
    \int_{s_{N_\varepsilon}}^{1}\!
   \left\{\frac{\psi''(0)}{2}
  \left[
  \left(\frac{\theta'(\xi_{0}) + \theta'(\xi_{N_\varepsilon})}{2}\right)^{2}
  -
  \theta'(s)^{2}
   \right]
   + O(\varepsilon) \right\}
  \,ds
  \right\rvert.
  \label{ee15}
\end{align}
Estimating the right-hand side in \eqref{ee15} can be done in a way similar to
the estimates that led from \eqref{ee9} to \eqref{ee13} and demonstrates that
\[|F_{\varepsilon}[\hat{\theta}_{\varepsilon}]-E_{0}[\theta]|=O(\varepsilon^2).\]



\smallskip
{\bf Proof of (b): Lower semicontinuity.}  We suppose $\theta\in H^1([0,1])$,
$\{\theta_\varepsilon\}\subset H^1([0,1])$,
and
$\theta_\varepsilon \rightharpoonup \theta$ in $H^1([0,1])$.  We
show that $\liminf_{\varepsilon\to 0}
F_\varepsilon[\theta_\varepsilon] \geq E_0[\theta]$.

If $\theta\notin H^1_{c}([0,1])$, then by Lemma~\ref{r1}
there is an $\bar{\varepsilon}>0$ such that
$\theta_\varepsilon \notin A_{\varepsilon}(0,1)$ and hence
$F_\varepsilon[\theta_\varepsilon]=\infty$
for all $0<\varepsilon\leq \bar{\varepsilon}$.
So we assume that $\theta\in H^1_{c}([0,1])$.  We can further assume that
$\theta_\varepsilon\in A_\varepsilon(0,1)$ for all $\varepsilon>0$.
Because $\{\theta_\varepsilon\}$ converges weakly in $H^1([0,1])$,
$\{\theta_\varepsilon'\}$ is bounded in $L^{2}([0,1])$.  Using \eqref{eq:theta_fun_der},
we see that there is a
constant $C$ such that for $j=1,\ldots,N_{\varepsilon}$
\begin{equation}
  \varepsilon^{-1}\left(
  \theta_{j}^{\varepsilon}-\theta_{j-1}^{\varepsilon}
  \right)^{2}
  \leq
  \sum_{i=1}^{N_{\varepsilon}}\varepsilon
  \left(\frac{\theta_{i}^{\varepsilon}-\theta_{i-1}^{\varepsilon}}{\varepsilon}\right)^{2}
  =
  \int_0^1(\theta_\varepsilon^\prime)^{2}\,ds
  \leq
  C  \label{e40}
\end{equation}
(recall that $\theta_i^\varepsilon=\theta_{\varepsilon}\left((2i-1)\varepsilon/2\right)$
for $i=1,\ldots,N_{\varepsilon}$ and that $\theta_0^\varepsilon=\theta_{N_{\varepsilon}}^\varepsilon-2\pi$).
It follows that
$\varepsilon\theta_\varepsilon'\leq (C\varepsilon)^{1/2}$ uniformly in $s$.
%
%
%
Therefore
\begin{equation}
  \varepsilon^{-2}\psi(\varepsilon\theta_{\varepsilon}')
  =
  \frac{1}{2}\psi''(0)(\theta_{\varepsilon}')^{2}+o(1)
  \label{e42}
\end{equation}
uniformly in $\varepsilon$ and $s\in[0,1]$.

Now we have
\begin{align}
  \liminf_{\varepsilon \rightarrow 0} \int_0^1 \psi_{\varepsilon}(\theta_{\varepsilon}')\,ds
  &=
  \liminf_{\varepsilon \rightarrow 0} \int_0^1
  \varepsilon^{-2}\psi(\varepsilon\theta_{\varepsilon}')\,ds \nonumber\\
  &=
  \liminf_{\varepsilon \rightarrow 0} \int_0^1\left(
  \frac{1}{2}\psi''(0)(\theta_{\varepsilon}')^{2}+o(1)\right)\,ds  \nonumber\\
  &=
  \liminf_{\varepsilon \rightarrow 0} \int_0^1
  \frac{1}{2}\psi''(0)(\theta_{\varepsilon}')^{2}\,ds  \nonumber\\
  &\geq
  \int_0^1 \frac{\psi''(0)}{2} (\theta')^2\,ds \label{e41},
\end{align}
where the last inequality follows from the weak lower semicontinuity of the
$L^{2}$ norm.
\end{proof}


%
%


\section*{Appendix 1. Proof of Lemma~\ref{r1}}
To simplify notation, we write just $\theta_{n}$ and $N_{n}$ for
$\theta_{\varepsilon_{n}}$ and $N_{\varepsilon_{n}}$.
Because $\theta_{n} \rightharpoonup \theta$ in $H^1([0,1])$, 
$\theta_{n}$ converges uniformly to $\theta$  on $[0,1]$ and hence 
$\cos \theta_{n}$ converges uniformly to $\cos \theta$ on $[0,1]$.
Thus
\begin{align}
  \int_{0}^{1}\!\cos \theta 
  &=
  \underset{n\rightarrow \infty}{\lim }
  \int_{0}^{1}\!\cos\theta_{n} \nonumber \\
  &=
  \underset{n\rightarrow \infty}{\lim }
  \left[
    \int_{0}^{{\varepsilon_{n}}/2}\!\cos\theta_{n}
    +
    \sum_{i=1}^{N_{n}-1} \int_{s_{i}}^{s_{i+1}}\!\cos\theta_{n}
    +
    \int_{1-\varepsilon_{n}/2}^{1}\!\cos\theta_{n}
  \right].
  \label{ee48}
\end{align}
The sequence $\{ \theta_{n}  \}$ is uniformly bounded, so the first
and third terms on the right-hand side of \eqref{ee48} go to zero 
with $\varepsilon_{n}$.  For the sum between these terms, we have
\begin{align}
  \sum_{i=1}^{N_{n}-1}
  \int_{s_{i}}^{s_{i+1}}\!\!\!\cos\theta_{n}(s)\,ds
  &=
  \sum_{i=1}^{N_{n}-1}
  \left[
  \int_{s_{i}}^{s_{i+1}}\!\!\!\cos\theta_{n}(s_{i})\,ds
  +
  \int_{s_{i}}^{s_{i+1}}\!\!\!
  \left(\cos\theta_{n}(s)-\cos\theta_{n}(s_{i})\right)\,ds
  \right] \nonumber \\
  &=
  \varepsilon\sum_{i=1}^{N_{n}-1}
  \cos\theta_{n}(s_{i})
  +
  \sum_{i=1}^{N_{n}-1}
  \int_{s_{i}}^{s_{i+1}}\!
  \left(\cos\theta_{n}(s)-\cos\theta_{n}(s_{i})\right)\,ds.
  \label{ee49}
\end{align}
On the right-hand side of \eqref{ee49}, the first sum is 0 because
$\theta_{n}\in A_{\varepsilon_{n}}(0,1)$ and the second sum is
easily shown to go to $0$ as $\varepsilon_{n}\rightarrow 0$ using the
uniform convergence of $\{ \theta_{n}  \}$ and the uniform
continuity of $\theta$
on $[0,1]$.
$\square$


\section*{Appendix 2. Proof of Lemma~\ref{r2}}
Note that $\tilde{\theta}(s):=\theta(s)-\theta(0)$ still satisfies
\eqref{ee45} and $\tilde{\theta}(0)=0$.  If $\tilde{\theta}^{*}$ is a
smooth function approximating $\tilde{\theta}$ in $H^1([0,1])$, then
$\tilde{\theta}^{*}+\theta(0)$ is a smooth function approximating
$\theta$ in $H^1([0,1])$.  Hence without loss of generality we assume
that $\theta(0)=0$.

Because $\theta$ is in $H^1([0,1])$, it is continuous.
So there exist $s_{1}, s_{2}, s_{3}, s_{4}\in (0,1)$ such that
$\pi/2<\theta(s)<\pi$ for $s_{1}<s<s_{2}$ and
$\pi<\theta(s)<3\pi/2$ for $s_{3}<s<s_{4}$.
By adding appropriately defined bump functions to $\theta(s)$,
we can produce functions $\theta_{1}$,
$\theta_{2}$, $\theta_{3}$, and $\theta_{4}$ each
close to $\theta$ in $H^1([0,1])$ such that
\begin{align}
  &\theta_{1}(s)=\theta(s)\  \text{for} \ s\notin (s_{1},s_{2})
  \ \text{and} \
  \pi/2<\theta_{1}(s)<\theta(s) \ \text{for} \ s_{1}<s<s_{2},
  \label{ee31} \\
  &\theta_{2}(s)=\theta(s) \ \text{for}\  s\notin(s_{3},s_{4})
  \ \text{and}\
  \pi<\theta_{2}(s)<\theta(s)\  \text{for} \ s_{3}<s<s_{4},  \label{ee36}\\
  &\theta_{3}(s)=\theta(s) \ \text{for}\  s\notin(s_{1},s_{2})
  \ \text{and}\
  \theta(s)<\theta_{3}(s)<\pi \  \text{for} \ s_{1}<s<s_{2},  \label{ee37}\\
  &\theta_{4}(s)=\theta(s) \ \text{for}\  s\notin(s_{3},s_{4})
  \ \text{and}\
  \theta(s)<\theta_{4}(s)<3\pi/2  \  \text{for}\ s_{3}<s<s_{4}. \label{ee38}
\end{align}
Each of these new functions still satisfies
$\theta_{i}(0)=0$ and 
$\theta_{i}(1)=2\pi$
since the outputs of $\theta$ need not be modified
near the endpoints of $[0,1]$.

Now we define $G[\vartheta]:=(\int_{0}^{1}\!\cos\vartheta(s)\,ds,
\int_{0}^{1}\!\sin\vartheta(s)\,ds)$ and 
$H\colon [0,1]^{2}\rightarrow \mathbb{R}^{2}$ by
\begin{align}
	H(\delta_{1},\delta_{2})
        &=
        (H_{1}(\delta_{1},\delta_{2}),H_{2}(\delta_{1},\delta_{2}))
        \nonumber \\
        &:=
        G[ \delta_{1}(\delta_{2}\theta_{1} + (1-\delta_{2})\theta_{4})
           +
           (1-\delta_{1})(\delta_{2}\theta_{2} + (1-\delta_{2})\theta_{3}) ].
	\label{ee18}
\end{align}
We show that
$H_{1}(0,\delta_{2})
=\int_{0}^{1}\!\cos(\delta_{2}\theta_{2}+(1-\delta_{2})\theta_{3})\,ds<0$ for $0\leq \delta_{2} \leq 1$.
To see this, we first observe that for $s\notin (s_{1},s_{2})\cup
(s_{3},s_{4})$,
$\delta_{2}\theta_{2}(s)+(1-\delta_{2})\theta_{3}(s)=\theta(s)$.
Next, if $s\in (s_{1},s_{2})$, then $\theta_{2}(s)=\theta(s)$
and $\theta_{3}(s)>\theta(s)$, so that
\begin{equation}
  \pi/2
  <
  \theta(s)
  <
  \delta_{2}\theta_{2}(s)+(1-\delta_{2})\theta_{3}(s)
  <
  \pi,
  \label{ee32}
\end{equation}
which in turn implies that
\begin{equation}
  \int_{s_{1}}^{s_{2}}\!\cos(\theta(s)) \,ds
  >
  \int_{s_{1}}^{s_{2}}\!\cos(\delta_{2}\theta_{2}(s)+(1-\delta_{2})\theta_{3}(s)) \,ds.
  \label{ee33}
\end{equation}
In a similar way we can show that
\begin{equation}
  \int_{s_{3}}^{s_{4}}\!\cos(\theta(s)) \,ds
  >
  \int_{s_{3}}^{s_{4}}\!\cos(\delta_{2}\theta_{2}(s)+(1-\delta_{2})\theta_{3}(s)) \,ds.
  \label{ee34}
\end{equation}
Hence
\begin{equation}
  H_{1}(0,\delta_{2})
  =
  \int_{0}^{1}\!\cos(\delta_{2}\theta_{2}(s)+(1-\delta_{2})\theta_{3}(s)) \,ds
  <
  \int_{0}^{1}\!\cos(\theta(s)) \,ds
  =
  0.
  \label{ee35}
\end{equation}
Because $H_{1}(0,\delta_{2})<0$ for $0\leq \delta_{2} \leq 1$ and
 $H_{1}(0,\cdot)$ is continuous on $[0,1]$, we can conclude that
$H_{1}(0,\delta_{2})\leq -\eta<0$ for $0\leq \delta_{2} \leq 1$ for
some $\eta>0$.
Similarly, we can show that
$H_{1}(1,\delta_{2})\geq \eta>0$ for $0\leq \delta_{2} \leq 1$
and that
$H_{2}(\delta_{1},0)\leq -\eta<0$,
$H_{2}(\delta_{1},1)\geq \eta>0$ for $0\leq \delta_{1} \leq 1$.

Next, for each $i$ we define
$\hat{\theta}_{i}(x):=\theta_{i}(x)-2\pi x$. 
Then $\hat{\theta}_{i}(0)=\hat{\theta}_{i}(1)$, 
and we can extend
$\hat{\theta}_{i}$ to a function on $\mathbb{R}$ with period 1.  
We use
convolution to approximate $\hat{\theta}_{i}$ in $H^{1}(\mathbb{R})$ by a smooth
function $\hat{\theta}^{*}_{i}$ that has period 1, so that in particular 
$\hat{\theta}^{*}_{i}(1)=\hat{\theta}^{*}_{i}(0)$.  We now define 
$\theta^{*}_{i}(x):=\hat{\theta}^{*}_{i}(x)+2\pi x$ and restrict
$\theta^{*}_{i}$ to $[0,1]$.  
Then  $\theta^{*}_{i}$ is a
smooth function that approximates $\theta_{i}$ in $H^{1}([0,1])$ and 
$\theta^{*}_{i}(1)=\theta^{*}_{i}(0)+2\pi$.  
Note that each
$\theta^{*}_{i}$ is uniformly close to $\theta_{i}$ on $[0,1]$.

Next we define $H^{*}$ as we defined $H$ in
\eqref{ee18} but replacing $\theta_{i}$ with $\theta^{*}_{i}$.
Because
$\delta_{2}\theta^{*}_{2}(s)+(1-\delta_{2})\theta^{*}_{3}(s)$
is uniformly close to
$\delta_{2}\theta_{2}(s)+(1-\delta_{2})\theta_{3}(s)$ for $s\in [0,1]$ and
$0\leq \delta_{2}\leq 1$, 
$H_{1}(0,\delta_{2})\leq -\eta<0$ for $0\leq \delta_{2} \leq 1$
implies that
$H^{*}_{1}(0,\delta_{2})<0$ for $0\leq \delta_{2} \leq 1$.
Likewise we have that
$H^{*}_{1}(1,\delta_{2})>0$ for $0\leq \delta_{2} \leq 1$
and that
$H^{*}_{2}(\delta_{1},0)<0$,
$H^{*}_{2}(\delta_{1},1)>0$ for $0\leq \delta_{1} \leq 1$.

Using the Intermediate Value Theorem, we see that
for each $0\leq \delta_{2} \leq 1$, there is a
$\hat{\delta}(\delta_{2})$ such that
$H^{*}_{1}(\hat{\delta}(\delta_{2}),\delta_{2})=0$.  Also,
$H^{*}_{2}(\hat{\delta}(0),0)<0$ and
$H^{*}_{2}(\hat{\delta}(1),1)>0$.  So there exists a
$\bar{\delta}$ such that
$H^{*}_{2}(\hat{\delta}(\bar{\delta}),\bar{\delta})=0$.  And
$H^{*}_{1}(\hat{\delta}(\bar{\delta}),\bar{\delta})=0$.
We define
$\theta^{*}=\hat{\delta}(\bar{\delta})(\bar{\delta}\theta^{*}_{1} + (1-\bar{\delta})\theta^{*}_{4})
           +
           (1-\hat{\delta}(\bar{\delta}))(\bar{\delta}\theta^{*}_{2}
           + (1-\bar{\delta})\theta^{*}_{3})$.
Then $\theta^{*} \in C^{2}([0,1])$, 
$\theta^{*}$ approximates $\theta$ in $H^{1}([0,1])$, 
and $\theta^{*}$ satisfies the constraints
\eqref{ee45}.
$\square$

\bibliography{gamma-converg-references}	

\end{document}